%% file: main.tex
\documentclass[letterpaper,10pt,twocolumn,twoside,conference]{IEEEtran}
\IEEEoverridecommandlockouts                         
\input{preamble.tex}
\usepackage[ruled, vlined]{algorithm2e}
\usepackage{mathrsfs}
\usepackage{hyperref}
\usepackage{svg}
\usepackage{paralist}

\newcommand{\Z}{\bb{Z}}
\theoremstyle{definition}

\title{\LARGE \textbf{Safety of Sampled-Data Systems with Control Barrier Functions \\via Approximate Discrete Time Models}}
\author{Andrew J. Taylor$^{1}$, Victor D. Dorobantu$^{1}$, Ryan K. Cosner$^{1}$, Yisong Yue, and Aaron D. Ames
\thanks{This research is supported by Ford, BP, and the National Science Foundation (CPS Award \#1932091, CMMI Award \#1923239).}
\thanks{$^1$Authors contributed equally. A. J. Taylor, V. D. Dorobantu, R. K. Cosner, Y. Yue, and A. D. Ames are with the Department of Computing and Mathematical Sciences, California Institute of Technology, Pasadena, CA 91125, USA, {\tt\small \{ajtaylor, vdoroban, rkcosner, yyue, ames\}@caltech.edu}. Y. Yue is affiliated with Argo AI, Pittsburgh, PA.
}
}

\newcommand{\XcapC}[1]{\mathcal{X} \cap (\mathcal{C} \oplus \overline{B}_{ #1 })}
\newcommand{\safety}{\mathtt{s}}

\begin{document}

\maketitle
\thispagestyle{empty}
\begin{abstract}
Control Barrier Functions (CBFs) have been demonstrated to be powerful tools for safety-critical controller design for nonlinear systems. Existing design paradigms do not address the gap between theory (controller design with continuous time models) and practice (the discrete time sampled implementation of the resulting controllers); this can lead to poor performance and violations of safety for hardware instantiations.  We propose an approach to close this gap by synthesizing sampled-data counterparts to these CBF-based controllers using approximate discrete time models and \emph{Sampled-Data Control Barrier Functions (SD-CBFs)}. Using properties of a system's continuous time model, we establish a relationship between SD-CBFs and a notion of \emph{practical safety} for sampled-data systems. Furthermore, we construct convex optimization-based controllers that formally endow nonlinear systems with safety guarantees in practice. We demonstrate the efficacy of these controllers in simulation. 
\end{abstract}

\section{Introduction}
\label{sec:intro}

Nonlinear control methods offer promising solutions to many modern safety-critical engineering applications. However, theoretically sound controller designs often fail to meet safety requirements when deployed on real systems. Thus, it is critical to understand the discrepancies between theoretical design and practical implementation mathematically, and to design controllers that close these gaps. Specifically, we address the challenges in designing safety-critical controllers for continuous time systems for which controllers are realized with discrete time sampling implementations, known as the sampled-data control problem \cite{monaco2007advanced}.

Control Barrier Functions (CBFs) have become a popular tool for constructively synthesizing controllers that endow nonlinear systems with rigorous safety guarantees \cite{ames2014control, ames2019control}. While originally posed for continuous time systems, they have similarly been developed for discrete time systems \cite{agrawal2017discrete} and sampled-data systems \cite{ghaffari2018safety, cortez2019control,gurriet2019realizable, yang2019selftriggered, singletary2020control, breeden2021control, usevitch2021adversarial, niu2021safety, zhang2021control}. These existing works take an \textit{emulative} approach to sampled-data control, in which continuous time safety conditions are met more conservatively to ensure that a system remains safe throughout a sample period. The approaches in \cite{cortez2019control, gurriet2019realizable, breeden2021control, usevitch2021adversarial, niu2021safety, zhang2021control} achieve this by adding a margin term to the standard CBF derivative condition that captures possible changes in the dynamics and CBF during the inter-sample period. This margin term often directly incorporates exponentials of Lipschitz constants and the sample period, requiring exceptionally high sample rates to achieve good performance, as studied in \cite{breeden2021control}. The work in \cite{singletary2020control} takes a computationally intensive approach to reduce conservatism by propagating sensitivity functions, which may be difficult for high-dimensional systems.

While the aforementioned results have focused on safety for sampled-data nonlinear systems, there exists a significant body of literature on stabilization of sampled-data nonlinear systems through discrete time design using \textit{approximate models}. Motivated by the challenge of finding exact representations of the discrete time sampled-data dynamics of nonlinear systems, the work in \cite{nevsic1999sufficient, nesic2004framework} proposed a framework for achieving a type of \textit{practical stability} using approximations of the discrete time sampled-data dynamics. Subsequently, a number of standard nonlinear stabilization techniques such as backstepping \cite{nesic2001backstepping}, model predictive control \cite{grune2003optimization}, Lyapunov-redesign \cite{nevsic2005lyapunov}, and optimization-based control via Control Lyapunov Functions \cite{taylor2021sampled}, were extended to use approximate models of discrete-time dynamics. These approaches often yielded significant improvements over their continuous time counterparts, even at relatively slow sample rates \cite{laila2003changing}. Notably, a similar framework for achieving safety has yet to be proposed.

In this work we propose a novel approach for achieving safety of sampled-data nonlinear systems via approximate models of discrete time sampled-data dynamics. In Section \ref{sec:background} we describe the sampled-data control setting and establish a \textit{consistency} result on how accurately sampled-data dynamics of a nonlinear system can be captured by a Runge-Kutta approximation. In Section \ref{sec:sdcbf} we propose a novel definition of \textit{practical safety} for sampled-data systems. Our definition mirrors the notion of practical stability developed in \cite{nevsic1999sufficient}, such that a system is practically safe if its state can be kept arbitrarily close to a safe set at sample times through sufficiently high sample rates. This leads to the unification of discrete time barrier functions \cite{agrawal2017discrete} with regularity properties developed in \cite{nevsic1999sufficient}, wherein we formulate \textit{Sampled-Data Barrier Functions (SD-BFs)} and their control counterparts, \textit{Sampled-Data Control Barrier Functions (SD-CBFs)}. We establish properties of this new class of CBFs and relate them to regular values.

The main contribution of this paper, given in Section \ref{sec:pracsafe}, establishes the practical safety of sampled-data systems through SD-BFs. We achieve this by connecting the key properties of SD-BFs with the accuracy guarantees provided by consistent approximations of the discrete sampled-data dynamics. This result is used to inform controller synthesis in Section \ref{sec:synthesis}, where we explore how appropriately designed SD-CBFs and Runge-Kutta approximations of systems with higher-order relative degrees preserve convexity with respect to the input of the SD-CBF difference constraint. This allows for SD-CBFs to be directly incorporated into a convex optimization-based controller that achieves practical safety. We demonstrate this controller in simulation, illustrating the relationship between sample rate and practical safety. The proof of the main result is given in the text, with all other proofs in the included appendix.

\section{Sampled-Data Control}
\label{sec:background}
Throughout this work, we will consider the nonlinear control system governed by the differential equation:
\begin{equation}\label{eqn:nonlinear-dynamics}
    \dot{\mb{x}} = \mb{f}(\mb{x}) + \mb{g}(\mb{x})\mb{u},
\end{equation}
for state signal $\mb{x}$ and control input signal $\mb{u}$ taking values in $\R^n$ and $\R^m$, respectively, drift dynamics $\mb{f}: \R^n \to \R^n$, and actuation matrix function $\mb{g}: \R^n \to \R^{n \times m}$. Consider an open subset $\mathcal{Z} \subseteq \R^n \times \R^m$ and its projection onto the state space \new{ $\mathcal{X} \triangleq \{\mb{x} \in \R^n ~|~ \exists~ \mb{u} \in \R^m ~\mathrm{s.t.}~ (\mb{x},\mb{u})\in\mathcal{Z}\}\subseteq\R^n$}. Assume there exists $T_{\mathrm{max}} \in \R_{++}$ \new{(the strictly positive reals)} such that for every state-input pair $(\mb{x}, \mb{u}) \in \mathcal{Z}$, there exists a unique solution $\bs{\varphi}:[0,T_{\mathrm{max}}]\to\R^n$ satisfying:
\begin{align}
\label{eqn:varphiode}
    \dot{\bs{\varphi}}(t) &= \mb{f}(\bs{\varphi}(t)) + \mb{g}(\bs{\varphi}(t))\mb{u}, \quad  %\quad \forall t \in (0, T_{\max}), \;\; 
    \bs{\varphi}(0) = \mb{x}. % \label{eqn:varphiic}
\end{align}
for all $t \in (0, T_\mathrm{max})$.
Given an $h\in(0,T_{\mathrm{max}}]$, we define a controller $\mb{k}: \mathcal{X} \to \R^m$ as \textit{$h$-admissible} if for any state $\mb{x} \in \mathcal{X}$, the state-input pair $(\mb{x}, \mb{k}(\mb{x}))$ satisfies $(\mb{x}, \mb{k}(\mb{x})) \in \mathcal{Z}$ and the corresponding solution $\bs{\varphi}$ satisfies $\bs{\varphi}(t) \in \mathcal{X}$ for all $t\in[0,h]$.
\begin{remark}
This requirement on $h$-admissible controllers will ensure that in the sampled-data context, the \new{closed-loop system is forward complete and its evolution may be described by iterative solutions to \eqref{eqn:varphiode}}. Though verifying $h$-admissibility of a controller may be intractable, assuming that a controller is $h$-admissible and renders the set $\mathcal{X}$ invariant is relatively weak as $\mathcal{X}$ is defined to ensure the continued existence of solutions rather than reflecting a task-specific set.
\end{remark}

The preceding construction of solutions and admissible controllers describes the sampled-data control setting, in which inputs are applied to the system with a zero-order hold over a sample period. More precisely, the set of possible sample periods is given by $I = (0, T_{\mathrm{max}}]$. Given a sample period $h\in I$ and an $h$-admissible controller $\mb{k}:\mathcal{X}\to\R^m$, the state and control input signals in \eqref{eqn:nonlinear-dynamics} satisfy:
\begin{equation} 
    \mb{u}(t) = \mb{k}(\mb{x}(t_k)) \quad \forall t\in[t_k,t_{k+1}),
\end{equation}
with sample times satisfying $t_{k+1}-t_k = h$ for all $k \in \bb{Z}_+$ \new{(the non-negative integers)}. In general, the evolution of the system over a sample period is given by the \textit{exact map} $\mb{F}_h^{e}: \mathcal{Z} \to \R^n$:
\begin{align}
    \mb{F}^{e}_h(\mb{x},\mb{u}) &= \mb{x} + \int_{0}^{h} [\mb{f}(\bs{\varphi}(\tau)) + \mb{g}(\bs{\varphi}(\tau))\mb{u}]~ \mathrm{d}\tau,
\end{align}
for all state-input pairs $(\mb{x}, \mb{u}) \in \mathcal{Z}$. We call $\{ \mb{k}_h: \mathcal{X} \to \R^m ~|~ h \in I \}$ a \textit{family of admissible controllers} if there is an $h^* \in I$ such that for each $h\in (0, h^*)$, $\mb{k}_h$ is $h$-admissible. This enables the following definition:

\begin{definition}[\textit{Exact Family}]
We define the \textit{exact family of maps} $\{\mb{F}^e_h~|~ h\in I\}$, and for a family of admissible controllers $\{ \mb{k}_h ~|~ h \in I \}$, we define the \textit{exact family of controller-map pairs} $\{ (\mb{k}_h, \mb{F}_h^{e}) ~|~ h \in I \}$.
\end{definition}
For all $h\in I$ such that $\mb{k}_h$ is $h$-admissible, the recursion $\mb{x}_{k + 1} = \mb{F}^{e}_h(\mb{x}_k, \mb{k}_h(\mb{x}_k))\in \mathcal{X}$ is well-defined for all $\mb{x}_0\in\mathcal{X}$ and $k\in\mathbb{Z}_+$. In practice, closed-form expressions for the exact family of maps are rarely obtainable, suggesting the use of approximations in the control synthesis process. While there are many approaches to approximating this family of maps, we will use the following common class of approximations:

\begin{definition}[\textit{Runge-Kutta Approximation Family}]
Let $p\in\mathbb{N}$. We define the \textit{Runge-Kutta approximation family of maps} $\{\mb{F}_h^{a,p}~|~ h\in I\}$, where for every sample period $h\in I$, define $\mb{F}^{a,p}_h:\mathcal{Z}\to\R^n$ recursively as:
\begin{align}
    \mb{F}_h^{a, p}(\mb{x}, \mb{u}) &= \mb{x} + h \sum_{i = 1}^p b_i (\mb{f}(\mb{z}_i) + \mb{g}(\mb{z}_i)\mb{u}), \label{eqn:rkupdate}\\
    \mb{z}_i &= \mb{x} + h \sum_{j = 1}^{i - 1} a_{i,j} (\mb{f}(\mb{z}_j) + \mb{g}(\mb{z}_j)\mb{u}), \label{eqn:rkz}
\end{align}
for all pairs $(\mb{x}, \mb{u}) \in \mathcal{Z}$, with $\mb{z}_1=\mb{x}$. Here, $b_1, \dots, b_p \in\R_{+}$ satisfy $\sum_{i =}^p b_i = 1$ and $a_{i,j}\in\R$ for each $i \in \{ 1, \dots p \}$ and $j \in \{1, \dots, i - 1 \}$. For a family of admissible controllers $\{\mb{k}_h~|~h\in I\}$, we may define the \textit{Runge-Kutta approximation family of controller-map pairs} $\{(\mb{k}_h,\mb{F}^{a,p}_h)~|~ h\in I\}$.
\end{definition}

\begin{remark}
For $h \in I$, the $h$-admissibility of $\mb{k}_h$ does not necessarily imply that the recursion $\mb{x}_{k + 1} = \mb{F}_h^{a, p}(\mb{x}_k, \mb{k}_h(\mb{x}_k))$ is well-defined for all $\mb{x}_0 \in \mathcal{X}$ and $k \in \Z_+$, though this is not necessary for our results.
\end{remark}

% \begin{remark}
% There may be an $h\in I$ such that the controller $\mb{k}_h$ is $h$-admissible but the recursion $\mb{x}_{k + 1} = \mb{F}^{a,p}_h(\mb{x}_k, \mb{k}_h(\mb{x}_k))$ is not well-defined for all $\mb{x}_0\in\mathcal{X}$ and $k\in\mathbb{Z}_+$. This is due to this map enabling $\mb{x}_k\notin\mathcal{X}$ for some $k>0$. While our results do not need this recursion to be well-defined, this can be achieved by extending the domain of $\mb{k}_h$ to $\R^n$. 
% \end{remark}

Defining class $\mathcal{K}$ ($\mathcal{K}_\infty$) and $\mathcal{K}^e$ ($\mathcal{K}^e_\infty$) comparison functions as in \cite{kellett2014compendium} and \cite{ames2019control}, the following definition characterizes how accurately an approximate map captures the exact map:

\begin{definition}[\textit{One-Step Consistency}]
 A family $\{ (\mb{k}_h,\mb{F}_h): h \in I \}$ is \textit{one-step consistent} with $\{(\mb{k}_h, \mb{F}_h^{e}) ~|~ h \in I \}$ over a set $A\subseteq\mathcal{X}$ if there exist a comparison function $\rho\in\K_\infty$ and $h^*\in I$ such that for all $\mb{x}\in A$ and $h\in(0,h^*)$, we have:
 \begin{equation}
 \label{eqn:one-step-cons}
     \Vert \mb{F}_h^{e}(\mb{x},\mb{k}_h(\mb{x}))-\mb{F}_h(\mb{x},\mb{k}_h(\mb{x})) \Vert \leq h\rho(h).
 \end{equation}
\end{definition}

Before establishing a relationship between a Runge-Kutta approximation family and one-step consistency, we state the following lemma we will use throughout this work:

\begin{lemma}
\label{lem:ballfit}
For any compact set $K \subset \mathcal{X}$, there is an $\varepsilon \in \R_{++}$ such that $K \oplus \overline{B}_\varepsilon \subset \mathcal{X}$ and $K \oplus \overline{B}_\varepsilon$ is compact, where $\overline{B}_\varepsilon$ is the closed norm-ball of radius $\varepsilon$ and $\oplus$ is the Minkowski sum.
\end{lemma}

We now provide our first result showing how properties of the dynamics and a family of controllers can be used to establish one-step consistency of the Runge-Kutta approximation family with the exact family of controller-map pairs:

\begin{theorem}
\label{thm:onestep}
Suppose $\mb{f}$ and $\mb{g}$ are locally Lipschitz continuous over $\mathcal{X}$. Let $K\subset\mathcal{X}$ be compact, consider a family of admissible controllers $\{\mb{k}_h ~|~ h \in I\}$, and suppose there exists $h_1\in I$ and a bound $M_{K}\in\R_{+}$ such that for every sample period $h\in(0,h_1)$, the controller $\mb{k}_h$ is bounded by $M_K$ over $K$. Then the family $\{(\mb{k}_h,\mb{F}^{a,p}_h)~|~ h\in I\}$ is one-step consistent with $\{(\mb{k}_h,\mb{F}^{e}_h)~|~ h\in I\}$ over the set $K$. 
\end{theorem}

\section{Sampled-Data Control Barrier Functions}
\label{sec:sdcbf}
In this section we develop a notion of \textit{practical safety} for sampled-data systems, and define \textit{Sampled-Data Control Barrier Functions} (SD-CBFs) as a tool for safety-critical sampled-data control synthesis. Lastly, we highlight a familiar setting which satisfies the properties required by SD-CBFs. 

We begin with the following definition relating the evolution of a sampled-data system and a set:
\begin{definition}[\textit{Forward Invariance}]
    A set $\mathcal{C}\subseteq \mathcal{X}$ is \textit{forward invariant} for a controller-map pair $(\mb{k}, \mb{F})$ if for every $\mb{x}_0 \in \mathcal{C}$ and number of steps $k \in \Z_+$, the recursion $\mb{x}_{k + 1} = \mb{F}(\mb{x}_k, \mb{k}(\mb{x}_k))$ is well-defined and satisfies $\mb{x}_k \in \mathcal{C}$.
\end{definition}

\begin{remark}
This definition of forward invariance requires that the system state be contained in the set $\C$ at sample times, which is aligned with the notion of stability for sampled-data systems presented in \cite{nevsic1999sufficient}. This differs from the standard definition of forward invariance used in the existing sampled-data safety literature, which additionally requires that the solution remain in the set $\C$ between sample times, i.e. $\bs{\varphi}(t)\in\C$ for $t\in[t_k,t_{k+1}]$. As seen in this literature, requiring inter-sample safety typically requires selecting control actions that meet a robustified continuous time barrier derivative condition. This robust condition typically depends on parameters of the system that are difficult to estimate, and using over-approximations may produce very conservative behavior \cite{breeden2021control}. Reducing this conservativeness usually amounts to operating at exceedingly high sample rates, which may not be practical, and which may excite unmodeled features of the system dynamics. Moreover, in practice, inter-sample safety violations at high sample rates can be inconsequential (and may not even be detectable). 
\end{remark}

We often do not have a closed-form expression for the exact family of maps and will need to design controllers using an approximate family of maps. The following definition will be used to describe the safety properties of the exact family of controller-map pairs when design uses approximations:

\begin{definition}[\textit{Practical Safety}]
A family $\{ (\mb{k}_h, \mb{F}_h)~|~ h \in I \}$ is \textit{practically safe} with respect to a set $\C \subseteq \mathcal{X}$ if for each $R \in \R_{++}$, there exists an $h^* \in I$ such that for each sample period $h \in (0, h^*)$, there is a corresponding set $\C_h \subseteq \mathcal{X}$ that is forward invariant for the controller-map pair $(\mb{k}_h, \mb{F}_h)$ and satisfies $\C \subseteq \C_h \subseteq \C \oplus \overline{B}_R$.
\end{definition}

\begin{remark}
This definition is posed to mirror that of practical stability for sampled-data systems proposed in \cite{nevsic1999sufficient}. In particular, the burden of proof lies with small values of $R$. If $R' \geq R$ and $\mathcal{C}_h$ is a forward invariant subset of $ \mathcal{C}\oplus \overline{B}_R$, then it is automatically a forward invariant subset of $\mathcal{C} \oplus \overline{B}_{R'}$ . 
\end{remark}

Before defining Sampled-Data Control Barrier Functions, for a non-empty set $\C\subseteq\mathcal{X}$, denote $d_{\C}(\mb{x}) = \inf_{\mb{y}\in\C}\Vert\mb{y}-\mb{x}\Vert$ for all $\mb{x} \in \mathcal{X}$. We now define Sampled-Data Barrier Functions and Sampled-Data Control Barrier Functions:

\begin{figure}[t]
    \centering
    \includegraphics[width=\linewidth]{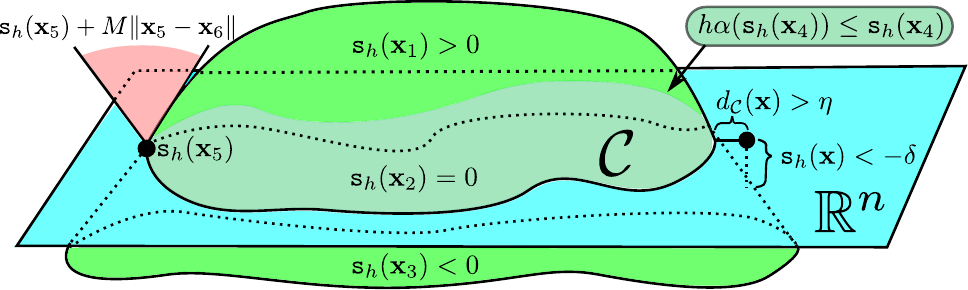}
    \caption{Visualizing the properties \eqref{eqn:bfvals}-\eqref{eqn:coercive} of SD-BF candidates. The dark green region represents the lower bound $h\alpha(\safety_h(\mb{x}_4))$ and $\safety_h(\mb{x}_6)$ cannot be in the red region due to the Lipschitz bound. }
    \label{fig:bfprops}
\end{figure}

\begin{definition}[\textit{Sampled-Data Barrier Function Candidate}]
\label{def:todobarfs}
Consider a set $\C\subseteq\mathcal{X}$. A collection of functions $\{ \safety_h ~|~ h \in I \}$ is a \textit{family of Sampled-Data Barrier Function Candidates} on $\C$ if there exist $h^* \in I$, a function $\alpha \in \mathcal{K}^e$, a radius $\varepsilon\in\R_{++}$, and a Lipschitz constant $M \in \R_{++}$ such that:
\begin{align}\
    & \safety_h(\mb{x}_1) > 0,\quad  \safety_h(\mb{x}_2) = 0, \quad \safety_h(\mb{x}_3) < 0,  \label{eqn:bfvals}\\
    % & \safety_h^{-1}(\R_+) = \C, \quad \safety_h^{-1}(0) = \partial \C\\
    & h\alpha(\safety_h(\mb{x}_4)) \leq \safety_h(\mb{x}_4), \label{eqn:class-k}\\
    &\vert \safety_h(\mb{x}_5) - \safety_h(\mb{x}_6) \vert \leq M \Vert \mb{x}_5 - \mb{x}_6 \Vert, \label{eqn:safelip}
\end{align}
for all states $\mb{x}_1 \in \mathrm{Int}(\C)$, $\mb{x}_2 \in \partial \C$, $\mb{x}_3 \in \mathcal{X} \setminus \C$, $\mb{x}_4\in \C$, $\mb{x}_5, \mb{x}_6\in \mathcal{X} \cap (\C \oplus \overline{B}_\varepsilon)$, and sample periods $h \in (0, h^*)$. Additionally, we require that for each $\eta \in \R_{++}$ there exists a $\delta \in \R_{++}$ such that\footnote{See Theorem \ref{thm:regval} for how this property relates to regular values.}:
\begin{equation}
\label{eqn:coercive}
    d_{\C}(\mb{x}) > \eta \implies \safety_h(\mb{x}) < -\delta,
\end{equation}
for all states $\mb{x} \in \mathcal{X} \cap (\C \oplus \overline{B}_\varepsilon)$ and sample periods $h \in (0, h^*)$.
\end{definition}
\begin{definition}[\textit{Sampled-Data Control Barrier Functions}]
A family of Sampled-Data Barrier Function Candidates $\{ \safety_h ~|~ h \in I \}$ is a \textit{family of Sampled-Data Control Barrier Functions} on $\mathcal{C}$ for $\{ \mb{F}_h ~|~ h \in I \}$ if for each state $\mb{x} \in \mathcal{X}$ and sample time $h \in (0, h^*)$, there exists a corresponding input $\mb{u} \in \R^m$ such that $(\mb{x}, \mb{u}) \in \mathcal{Z}$ and:
\begin{equation}
    \safety_h(\mb{F}_h(\mb{x}, \mb{u})) - \safety_h(\mb{x}) \geq -h \alpha (\safety_h(\mb{x})).\label{eqn:dt_cbf_condition}
\end{equation}
\end{definition}
\begin{definition}[\textit{Sampled-Data Barrier Function}]
Given a family of admissible controllers $\{ \mb{k}_h ~|~ h \in I \}$, a family of Sampled-Data Barrier Function Candidates $\{ \safety_h ~|~ h \in I \}$ is a \textit{family of Sampled-Data Barrier Functions} on $\mathcal{C}$ for $\{ (\mb{k}_h, \mb{F}_h) ~|~ h \in I \}$ if:
\begin{equation}
    \safety_h(\mb{F}_h(\mb{x}, \mb{k}_h(\mb{x}))) - \safety_h(\mb{x}) \geq -h \alpha (\safety_h(\mb{x})),\label{eqn:dt_bf_condition}
\end{equation}
for all states $\mb{x} \in \mathcal{X}$ and sample times $h \in (0, h^*)$.
\end{definition}

\begin{remark}
We note that the conditions in \eqref{eqn:bfvals} and \eqref{eqn:class-k} are standard conditions required by barrier functions for discrete systems \cite{agrawal2017discrete}. The inequalities in \eqref{eqn:bfvals} imply that for each $h\in(0,h^*)$, $\C$ is the 0-superlevel set of $\safety_h$. The inequality in \eqref{eqn:class-k} places a requirement on the SD-BF decrement through \eqref{eqn:dt_bf_condition} that implies that for each $h\in(0,h^*)$, $\C$ is forward invariant for $(\mb{k}_h,\mb{F}_h)$. The condition in \eqref{eqn:safelip} requires the SD-BF to be Lipschitz continuous over a slightly larger set than $\C$ with a Lipschitz constant that is uniform in the sample period, and will be used to relate exact and approximate families through one-step consistency. The implication in \eqref{eqn:coercive} resembles a \emph{coercivity} condition, requiring the SD-BF value to decrease locally outside of the set $\C$ in a way that is uniform in the sample period. This property will be critical for producing forward invariant sets contained in $\C\oplus\overline{B}_R$ for arbitrarily small values of $R$. The distinction between the conditions in \eqref{eqn:dt_cbf_condition} and \eqref{eqn:dt_bf_condition} is that the former condition states the possibility of safe control synthesis for an open-loop system, while the latter applies as a \textit{certificate} for a closed-loop system. These properties are illustrated in Fig. \ref{fig:bfprops}.
\end{remark}

To more clearly understand the nature of the properties \eqref{eqn:bfvals}-\eqref{eqn:dt_cbf_condition} we will discuss a familiar setting in which they are implied. As frequently used in the continuous time Control Barrier Function literature \cite{ames2019control}, a continuously differentiable function $\safety: \mathcal{X} \to \R$ has $c \in \R$ as a \textit{regular value} if $\safety(\mb{x}) = c$ implies $\nabla \safety(\mb{x}) \neq \mb{0}$ for all states $\mb{x} \in \mathcal{X}$. The following result connects regular values and the property in \eqref{eqn:coercive}:

\begin{theorem}
\label{thm:regval}
Suppose that $\safety : \mathcal{X} \to \R$ is twice continuously differentiable with a compact $0$-superlevel set $\C$ and $0$ as a regular value. If $d_\C$ is defined using the 2-norm, then there is an $\varepsilon \in \R_{++}$ such that each $\eta \in \R_{++}$ corresponds to a $\delta \in \R_{++}$ satisfying:
\begin{align}
    d_{\C}(\mb{x}) > \eta \implies \safety (\mb{x}) < - \delta. 
\end{align}
for all states $\mb{x} \in\XcapC{\varepsilon}$. 
\end{theorem}

\section{Practical Safety}
\label{sec:pracsafe}
In this section we present our main contribution by establishing how a family of SD-BFs for an approximate family can be used to ensure the practical safety of the exact family of controller-map pairs via one-step consistency:
\begin{theorem}\label{thm:pracsafe}
Consider a set $\C\subseteq\mathcal{X}$ and a family of admissible controllers $\{ \mb{k}_h ~|~ h \in I \}$. Suppose that:
\begin{enumerate}
    \item There exists a family of Sampled-Data Barrier Functions on $\C$ for a family $\{ (\mb{k}_h, \mb{F}_h) ~|~ h \in I \}$.
    \item There exists an $\varepsilon'\in\R_{++}$ such that the family $\{ (\mb{k}_h, \mb{F}_h) ~|~ h \in I \}$ is one-step consistent with the exact family $\{ (\mb{k}_h, \mb{F}^e_h) ~|~ h \in I \}$ over the set $\mathcal{X}\cap(\C\oplus\overline{B}_{\varepsilon'})$.
\end{enumerate}
Then the exact family $\{ (\mb{k}_h, \mb{F}^e_h) ~|~ h \in I \}$ is practically safe with respect to $\C$.
\end{theorem}

\begin{proof}
Let $h_1^*$, $\alpha$, $\varepsilon$, and $M$ be defined as in Definition \ref{def:todobarfs}. By assumption, there exists an $h_2^*\in I$ and $\rho\in\mathcal{K}_\infty$ such that \eqref{eqn:one-step-cons} holds for all $\mb{x}\in\mathcal{X}\cap(\C\oplus\overline{B}_{\varepsilon'})$ and $h\in(0,h_2^*)$. Since the family of controllers is assumed to be admissible, there is an $h_3^* \in I$ such that $\mb{k}_h$ is $h$-admissible for each $h \in (0, h_3^*)$.

Let $R \in \R_{++}$, and pick $R'\in \R_{++}$ such that $R' \leq \min \{\varepsilon, \varepsilon', R\}$. By \eqref{eqn:coercive}, there exist $\delta,\Delta \in \R_{++} $ such that: 
\begin{align}
    d_\mathcal{C}(\mb{x}) > R'/2 &\implies s_h(\mb{x}) < - \delta , \\\label{eqn:bigboi}
    d_\mathcal{C}(\mb{x}) > \delta/(2M) &\implies s_h(\mb{x}) < - \Delta, 
\end{align}
for all $\mb{x} \in \XcapC{\varepsilon}$ and $h\in(0,h_1^*)$. Fix $h \in I$ with $h < \min{\{ h_1^*, h_2^*, h_3^* \}}$. For any $c \in \R$, we denote the $c$-superlevel set of $\safety_h $ as:
\begin{align}
    \Omega_{c,h} =  \{ \mb{x} \in \mathcal{X} ~|~ \safety_h(\mb{x}) \geq c\}.
\end{align}
For any state $\mb{x} \in \Omega_{-\delta, h}$, we have $d_\C(\mb{x}) \leq R' / 2$, and thus $\C \subseteq \Omega_{-\delta, h}\subseteq \XcapC{R'/2} \subseteq \mathcal{C} \oplus \overline{B}_R$.

We will prove that for small enough $h$, the set $\Omega_{-\delta,h}$ is forward invariant for the controller-map pair $(\mb{k}_h,\mb{F}_h^e)$. We denote three cases (see Fig. \ref{fig:sweet_bean}), considering a state $\mb{x} \in \mathcal{X}$ such that either \begin{inparaenum}[\bfseries (1)] \item $\mb{x} \in \mathcal{C}$, \item $\mb{x} \in \Omega_{-\delta, h }\setminus\C $ and $d_\mathcal{C}(\mb{x}) \leq \delta/(2M) $, or \item $\mb{x} \in \Omega_{-\delta, h }\setminus\C$ and $d_\mathcal{C}(\mb{x}) > \delta/(2M)$. \end{inparaenum}

\textbf{Case 1: }
Suppose $\mb{x} \in \mathcal{C}$. From \eqref{eqn:dt_bf_condition} and \eqref{eqn:class-k}, we have:
\begin{equation}
    \safety_h(\mb{F}_h(\mb{x}, \mb{k}_h(\mb{x}))) - \safety_h(\mb{x}) \geq -h\alpha(\safety_h(\mb{x})) \geq -\safety_h(\mb{x}),
\end{equation}
so $\safety_h(\mb{F}_h(\mb{x}, \mb{k}_h(\mb{x}))) \geq 0$, or $\mb{F}_h(\mb{x}, \mb{k}_h(\mb{x})) \in \mathcal{C}$. By one-step consistency, we have:
\begin{equation}
    \| \mb{F}_h^e(\mb{x}, \mb{k}_h(\mb{x})) - \mb{F}_h(\mb{x}, \mb{k}_h(\mb{x})) \| \leq h\rho(h),
\end{equation}
so if $h\rho(h) \leq \varepsilon$, then $\mb{F}_h^e(\mb{x}, \mb{k}_h(\mb{x}))\in \XcapC{\varepsilon}$. Thus:
\begin{equation}
    | \safety_h(\mb{F}_h^e(\mb{x}, \mb{k}_h(\mb{x}))) - \safety_h(\mb{F}_h(\mb{x}, \mb{k}_h(\mb{x}))) | \leq Mh\rho(h),
\end{equation}
and if $Mh\rho(h) \leq \delta$ as well, then:
\begin{equation}
    \safety_h(\mb{F}_h^e(\mb{x}, \mb{k}_h(\mb{x}))) \geq \safety_h(\mb{F}_h(\mb{x}, \mb{k}_h(\mb{x}))) - Mh\rho(h) \geq -\delta,
\end{equation}
giving us $\mb{F}_h^e(\mb{x}, \mb{k}_h(\mb{x})) \in \Omega_{-\delta, h}$. The analysis of this case gives us the requirement $h\rho(h) \leq \min\{ \varepsilon, \delta / M \}$.

\begin{figure}[b]
    \centering
    \includegraphics[width=\linewidth]{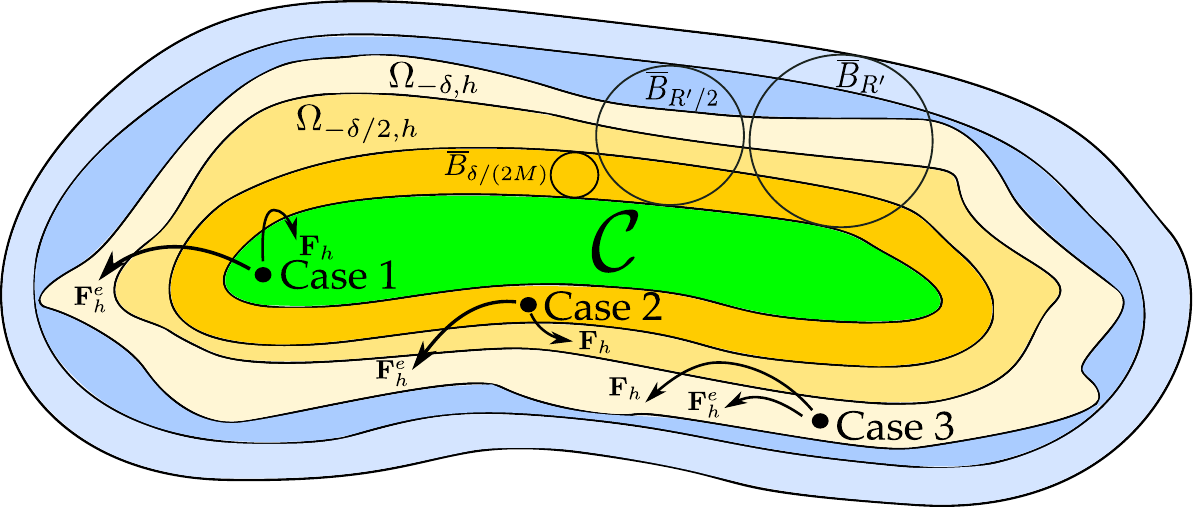}
    \caption{A visual representation of the main sets and three cases discussed in the proof of Theorem \ref{thm:pracsafe}.} 
    \label{fig:sweet_bean}
\end{figure}

Before continuing to cases 2 and 3, we establish some additional properties. First, note that the superlevel sets have the containment property  $\Omega_{-\delta/2, h} \subseteq \Omega_{-\delta,  h}$. Next, for any $\eta \in\R_{++} $ and any $\mb{x} \in \XcapC{\varepsilon}$ with $\mb{x} \notin \mathcal{C}$, there is a state $\mb{y} \in\mathcal{C}$ such that $\Vert \mb{x} - \mb{y} \Vert  < d_{\mathcal{C}}(\mb{x}) + \eta$. Therefore: 
\begin{align}
    \safety_h(\mb{x}) \geq \safety_h(\mb{y}) - M\Vert \mb{x} - \mb{y} \Vert  \geq - Md_\mathcal{C}(\mb{x}) - M\eta,
\end{align}
since $\safety_h(\mb{y}) \geq 0 $. Since $\eta$ can be chosen arbitrarily small, we have $\safety_h(\mb{x}) \geq -Md_\C(\mb{x})$. If $d_\mathcal{C}(\mb{x}) \leq \delta/(2M) $, then $\safety_h(\mb{x}) \geq -\delta/2$, so $\XcapC{\delta/(2M) } \subseteq \Omega_{-\delta/2, h}\subseteq \Omega_{-\delta,h}$. 

Next, consider $\mb{x} \in \Omega_{-\delta,h} \setminus \C $. Since $\mb{x} \notin\C$, meaning $\safety_h(\mb{x}) < 0$ and thus $\alpha(\safety_h(\mb{x})) < 0$, we have from \eqref{eqn:dt_bf_condition} that: \begin{align}
    \safety_h(\mb{F}_h(\mb{x}, \mb{k}_h(\mb{x}))) \geq \safety_h(\mb{x}) - h\alpha(\safety_h(\mb{x}) )  > - \delta. 
\end{align}
Thus $\mb{F}_h(\mb{x}, \mb{k}_h(\mb{x})) \in \Omega_{-\delta, h }  \subseteq \XcapC{R'/2}$ so we can apply one step consistency to achieve:
\begin{align}
    \Vert \mb{F}^e_h(\mb{x}, \mb{k}_h(\mb{x})) - \mb{F}_{h}(\mb{x}, \mb{k}_h(\mb{x}))\Vert \leq h\rho(h).
\end{align}
If $h \rho(h) \leq R'/2$, then $\mb{F}_h^e(\mb{x}, \mb{k}_h(\mb{x})) \in \XcapC{R'}$, in which case the Lipschitz property of $\safety_h$ yields the bound: 
\begin{align}\label{eqn:forward-lip}
    \vert \safety_h(\mb{F}^e_h(\mb{x}, \mb{k}_h(\mb{x}) )) - \safety_h(\mb{F}_h(\mb{x}, \mb{k}_h(\mb{x}) )) \vert \leq M h \rho (h).
\end{align}
Note that because $R' / 2 < \varepsilon$, the requirement from Case 1 can be replaced by $h\rho(h) \leq \min{\{ R' / 2, \delta / M \}}$.

\textbf{Case 2: }
Suppose $\mb{x} \in \Omega_{-\delta, h}\setminus\C$ and $d_{\mathcal{C}}(\mb{x}) \leq \delta/(2M)$. Since $\mb{x} \not\in \C$ and $\XcapC{\delta / (2M)} \subseteq \Omega_{-\delta / 2, h}$, we have $-\delta / 2 \leq \safety_h(\mb{x}) < 0$. Therefore:
\begin{equation}
    \safety_h(\mb{F}_h(\mb{x}, \mb{k}_h(\mb{x}))) \geq \safety_h(\mb{x}) - h\alpha(\safety_h(\mb{x}) ) \geq -\delta / 2,
\end{equation}
so $\mb{F}_h(\mb{x}, \mb{k}_h(\mb{x})) \in \Omega_{-\delta / 2, h}$. By adding and subtracting $\safety_h(\mb{F}_h^e(\mb{x}, \mb{k}_h(\mb{x})))$ and using \eqref{eqn:forward-lip}, we have:
\begin{equation}
    \safety_h(\mb{F}_h^e(\mb{x}, \mb{k}_h(\mb{x}))) \geq -Mh\rho(h) - \delta / 2,
\end{equation}
when $h\rho(h) \leq R' / 2$. If $Mh\rho(h) \leq \delta / 2$ as well, then $\safety_h(\mb{F}_h^e(\mb{x}, \mb{k}_h(\mb{x}))) \geq -\delta$, or $\mb{F}_h^e(\mb{x}, \mb{k}_h(\mb{x}))\in \Omega_{-\delta, h}$. Thus we update the requirements to be $h\rho(h) \leq \min{\{ R' / 2, \delta / (2M) \}}$.

\begin{figure*}[t]
    \centering
    \includegraphics[width=0.95\linewidth]{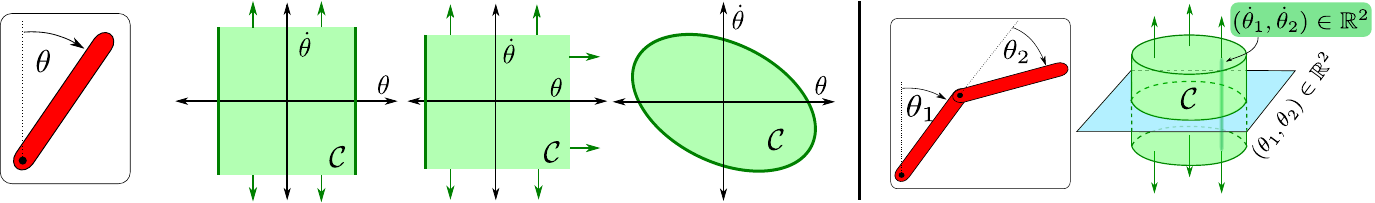}
    \caption{The single (left) and double (right) inverted pendulums and their safe sets. The inverted pendulum safe sets are, from left to right, a configuration ellipsoid, a halfspace, and a Lyapunov sublevel set. The far right image shows the safe set for the double inverted pendulum where $(\theta_1, \theta_2)\in\R^2$ are constrained to an ellipse and no constraint is placed on $(\dot{\theta}_1, \dot{\theta}_2) \in \R^2 $ which are represented by vertical fibers. }
    \label{fig:safesets}
\end{figure*}

\begin{table*}[t]
    \centering
    \begin{tabular}{|c|c|c|c|c|}
    \hline 
        System & $\mb{q}$ & $\mb{D}: \R^m \to \mathbb{S}^m_{++}$ & $\mb{C}: \R^m \times \R^m \to \R^{m\times m }$ & $\mb{G}:\R^m \to \R^m $  \\
    \hline
        Single & $\theta$ & 1 & 0 & $-\sin\theta$\\
    \hline 
        Double & $\begin{bmatrix} \theta_1 \\ \theta_2 \end{bmatrix}$ &
        $\lmat 3 + 2\cos\theta_2 & 1+ \cos\theta_2  \\
        1 + \cos\theta_2 & 1 \rmat$
        & 
        $\lmat 0 & -(2\dot{\theta}_1 + \dot{\theta}_2)\sin\theta_2\\
        \frac{1}{2}(2\dot{\theta}_1 +  \dot{\theta}_2)\sin\theta_2
        & -\frac{1}{2} \dot{\theta}\sin\theta_2
        \rmat 
        $
        & 
        $
        \lmat -2\sin\theta_1 - \sin(\theta_1 + \theta_2) \\
        -\sin(\theta_1 + \theta_2) 
        \rmat $\\
    \hline
    \end{tabular}
    \caption{Terms in pendulum dynamics given by \ref{eqn:lagrange_dyn}. Angles are taken clockwise from upright, and $\theta_2$ is taken relative to $\theta_1$.}
    \label{fig:results}
    \vspace{-0.6cm}
    \label{table:dynamics}
\end{table*}

\textbf{Case 3: }
Suppose $\mb{x} \in \Omega_{-\delta, h}\setminus\C$ and $d_\mathcal{C}(\mb{x}) > \delta/(2M)$. From \eqref{eqn:bigboi}, we have:
\begin{equation}
    \safety_h(\mb{F}_h(\mb{x}, \mb{k}_h(\mb{x}))) - \safety_h(\mb{x}) > -h\alpha(-\Delta).
\end{equation}
Adding and subtracting $\safety_h(\mb{F}_h^e(\mb{x}, \mb{k}_h(\mb{x})))$ and \eqref{eqn:forward-lip} yield:
\begin{align}
    \safety_h(\mb{F}_h^e(\mb{x}, \mb{k}_h(\mb{x}))) &> \safety_h(\mb{x}) -Mh\rho(h) - h\alpha(-\Delta), \\
    &= \safety_h(\mb{x}) - h(M\rho(h) + \alpha(-\Delta)),
\end{align}
when $h\rho(h) \leq R' / 2$. If $M\rho(h) \leq -\alpha(-\Delta)$ as well, then $\safety_h(\mb{F}_h^e(\mb{x}, \mb{k}_h(\mb{x}))) > \safety_h(\mb{x}) \geq -\delta$, or $\mb{F}_h^e(\mb{x}, \mb{k}_h(\mb{x})) \in \Omega_{-\delta, h}$. 

To conclude, if both:
\begin{enumerate}
    \item $h < \min{\left\{ h_1^*, h_2^*, h_3^*, \rho^{-1}( -\alpha(-\Delta) / M ) \right\}}$,
    \item $h\rho(h) \leq \min{\{ R' / 2, \delta / (2M) \}}$,
\end{enumerate}
then the set $\C_h\triangleq\Omega_{-\delta, h}\subseteq \C\oplus\overline{B}_R$ is forward invariant for the controller-map pair $(\mb{k}_h,\mb{F}^{e}_h)$, and thus the family $\{(\mb{k}_h,\mb{F}^e_h)~|~h\in I\}$ is practically safe with respect to $\C$.
\end{proof}

\section{Control Synthesis \& Simulation}
\label{sec:synthesis}
In this section we explore convexity of the CBF decrement condition, and define an optimization-based controller via an SD-CBF for achieving practical safety, which we deploy in simulation on inverted and double inverted pendulums.

The following result establishes how for a system with a block integrator structure, a Runge-Kutta approximation family of maps of the appropriate order can preserve a convexity property of a family $\{\safety_h ~|~ h\in I\}$:

% This will ensure the inequality \eqref{eqn:dt_cbf_condition} is convex with respect to the input, allowing it to be used in a convex optimization based controller.

\begin{theorem} Consider $\ell, \gamma, q \in \N$ such that $n = \ell  \gamma$ and $q\leq\gamma$. Suppose the system dynamics have the form: \label{thm:convexity}
\begin{equation}
    \dot{\mb{x}} = \underbrace{\begin{bmatrix} 
        \mb{0} & \mb{I} && \\
        & \ddots & \ddots\\
        && \mb{0} & \mb{I} \\
        &&& \mb{0}
    \end{bmatrix}}_{\mb{A}} \mb{x} + \underbrace{\begin{bmatrix} \mb{0} \\ \vdots \\ \mb{0} \\ \mb{f}_\gamma(\mb{x}) + \mb{g}_\gamma(\mb{x})\mb{u}  \end{bmatrix}}_{\mb{r}(\mb{x},\mb{u})}, \label{eq:integrator_dyn}
\end{equation}
where $\mb{f}_\gamma: \R^n \to \R^\ell$ and $\mb{g}_\gamma: \R^n \to \R^{\ell \times m}$. For each $h \in I$, consider a function $\safety_h: \R^n \to \R$, and suppose there exists a function $\tilde{\safety}_h: (\R^\ell)^q \to \R$ satisfying:
\begin{equation}
\label{eqn:stilde}
    \safety_h(\mb{x}) = \tilde{\safety}_h(\bs{\zeta}_1, \dots, \bs{\zeta}_q),
\end{equation}
for all $\mb{x} = (\bs{\zeta}_1, \dots, \bs{\zeta}_\gamma) \in (\R^{\ell})^\gamma \simeq \R^n$. If the function $\tilde{\safety}_h$ is concave with respect to its last argument and $p = \gamma - q+1$, then for $\alpha \in \mathcal{K}^e$, the function $\phi_h:\mathcal{Z}\to\R$ defined as:
\begin{align}
\label{eqn:cvxcnstrnt}
    \phi_h(\mb{x},\mb{u}) = -\safety_h(\mb{F}_h^{a,p}(\mb{x}, \mb{u})) + \safety_h(\mb{x}) - h \alpha(\safety_h(\mb{x})),
\end{align} 
is convex in its second argument.
\end{theorem}

Safety-critical controllers are frequently synthesized using Control Barrier Functions and convex optimization (typically quadratic programs) \cite{ames2019control}. The following result highlights how we may similarly synthesize a family of controllers that achieve practical safety through optimization:

\begin{theorem}
\label{thm:optprob}
Let $\{\safety_h ~|~ h \in I \}$ be a family of SD-CBFs on $\C$ for a family $\{\mb{F}^{a,p}_h~|~ h\in I\}$ such that the set:
\begin{equation}
    \mathcal{F}(\mb{x}) = \{\mb{u}\in\R^m ~|~ (\mb{x},\mb{u})\in\mathcal{Z} ~\mathrm{and}~ \phi_h(\mb{x},\mb{u})\leq 0\},
\end{equation}
is closed and convex for each $h\in I$ and $\mb{x}\in\mathcal{X}$. Consider a set of controllers $\{\mb{k}_h~|~ h\in I\}$ satisfying:
\begin{align}
    \mb{k}_h (\mb{x}) = \argmin_{\mb{u} \in \R^m } &~ \frac{1}{2}\Vert \mb{u}-\mb{k}_d(\mb{x}) \Vert_2^2 \tag{SD-CBF-OP} \label{eq:sdcbf Controller} \\
    \mathrm{s.t. } &~ \safety_h(\mb{F}_h^{a,p}(\mb{x}, \mb{u})) - \safety_h (\mb{x}) \geq -h \alpha (\safety_h (\mb{x}) ), \nonumber
\end{align}
for each $\mb{x}\in\mathcal{X}$ and $h\in (0, h^*)$, where $\mb{k}_d:\mathcal{X}\to\R^m$ is a nominal controller. If $\{\mb{k}_h~|~ h\in I\}$ is a family of admissible controllers, then $\{\safety_h ~|~ h\in I\}$ is a family of Sampled-Data Barrier Functions on $\C$ for $\{(\mb{k}_h,\mb{F}^{a,p}_h)~|~ h\in I\}$.
\end{theorem}

We use this controller in simulation on fully-actuated single and double inverted pendulums, with dynamics given by: 
\begin{align}
    \mb{D}(\mb{q}) \ddot{\mb{q}} + \mb{C}(\mb{q}, \dot{\mb{q}}) \dot{\mb{q}} + \mb{G}(\mb{q}) & = \mb{u} \label{eqn:lagrange_dyn}
\end{align}
where $\mb{D}$, $\mb{C}$, and $\mb{G}$ are functions encoding inertia, Coriolis, and gravity terms, $\mb{q}, \dot{\mb{q}} \in \R^m$ are configuration and velocity vectors, and $\mb{u} \in \R^m$ is a torque vector. These terms are detailed in Table \ref{table:dynamics}. With state vector $\mb{x} = (\mb{q}, \dot{\mb{q}}) \in \R^n$, the dynamics in \eqref{eqn:lagrange_dyn} can be expressed in the form \eqref{eq:integrator_dyn}, where $\ell = m$ and $\gamma = 2$. For the single inverted pendulum we use safe sets with the form of a Lyapunov sublevel ($\safety_h(\mb{x}) = 1 - \mb{x}^\top\mb{P}\mb{x}$ with $\mb{P}\in\mathbb{S}^2_{++}$ solving the continuous algebraic Riccati equation with feedback linearized dynamics, state cost matrix $\mb{I}_2$, and input cost matrix $\mb{I}_1$), a configuration ellipsoid ($\tilde{\safety}_h(\theta) = 1 - \theta^2$), and a halfspace ($\tilde{\safety}_h(\theta) = \theta + 0.1$). For the double inverted pendulum we enforce safety of a configuration ellipsoid ($\tilde{\safety}_h(\mb{q}) = 1 -  \|\mb{q}\|_2^2$). These sets are visualized in Fig.~\ref{fig:safesets}. We use Runge-Kutta approximations with $p=1$ (forward Euler) for the Lyapunov sublevel set and $p=2$ (midpoint rule) for the other settings. Controllers of the form \eqref{eq:sdcbf Controller} are employed with identity comparison functions; for the Lyapunov sublevel set, $\mb{k}_d$ is a feedback linearizing controller with auxiliary PD control (proportional gain $1$, derivative gain $2$), and for the other settings, $\mb{k}_d$ is a zero (constant) controller. With $11$ sample periods spaced logarithmically (over $[0.05, 0.5]$ and $[0.01, 0.1]$ seconds for the single and double inverted pendulums, respectively) and initial conditions sampled from each safe set, the closed-loop systems are simulated for $10$ seconds. For the inverted pendulum, $500$ initial states are sampled uniformly from the Lyapunov sublevel set, and $41 \times 41$ grids of initial states cover $[-1, 1] \times [-5, 5]$ for the configuration ellipsoid and $[-0.1, 1] \times [-5, 5]$ for the halfspace. For the double inverted pendulum, $500$ initial states are uniformly sampled with configurations in the unit Euclidean ball in $\R^2$ and velocities in  $[-1, 1]^2$. The worst-case distances from the safe sets are reported as a function of sample period in Fig. \ref{fig:sim_plot}. These distances decrease for sufficiently small sample periods.
\begin{figure}[h]
    \centering
    \includegraphics[trim={0in, 3in, 0in, 3in}, clip, scale =  0.41]{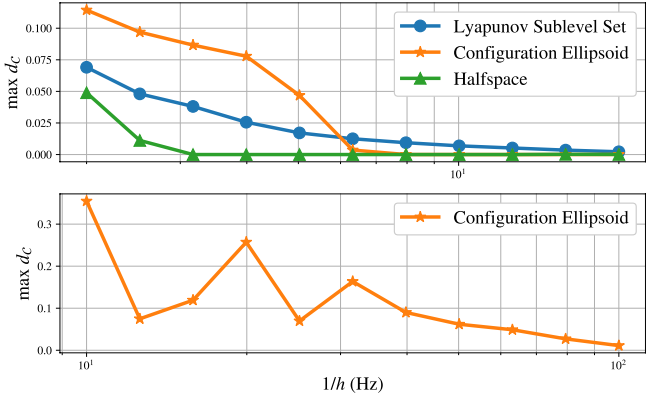}
    \caption{ The maximum distance from the safe set $\mathcal{C}$ (lower is better) achieved during trials vs. the the sampling frequency. The simulations and corresponding animations be found at \url{https://bit.ly/CBF-OP} and \url{https://vimeo.com/690803272}. \textbf{Top: } The inverted pendulum for 3 different safe sets. \textbf{Bottom: }  The double inverted pendulum.}
    \label{fig:sim_plot}
\end{figure}

\section{Conclusion}
In this work we have developed a novel approach for safety-critical sampled-data control through approximate discrete time models. Our main contribution, Sampled-Data Control Barrier Functions, provides a tool for designing practically safe controllers. Future work will study the relationship between approximation maps and control methods like backstepping.

\appendix

\subsection{Proof of Lemma \ref{lem:ballfit}}
\begin{proof}
As $\mathcal{X}$ is open, for every $\mb{x} \in K$, there is a corresponding open ball centered at $\mb{x}$ with radius $\delta_{\mb{x}} \in \R_{++}$ that is contained in $\mathcal{X}$. Let $B_{\mb{x}} \subset \mathcal{X}$ be the open ball centered at $\mb{x}$ of radius $\delta_{\mb{x}} / 2$. Consider the collection $\{B_{\mb{x}}: \mb{x} \in K \}$; this is an open cover for the compact set $K$, so some finite collection $B_{\mb{x}_1}, \dots, B_{\mb{x}_N}$ for some $\mb{x}_1, \dots, \mb{x}_N \in K$, respectively, also covers $K$. Let $\delta = \min_i \delta_{\mb{x}_i}$, and consider any $\mb{z} \in K \oplus \overline{B}_{\delta / 4}$. There is some $\mb{x} \in K$ such that $\| \mb{z} - \mb{x} \| \leq \delta / 4$ and some $i \in \{ 1, \dots, N \}$ such that $\| \mb{x} -  \mb{x}_i \| < \delta_{\mb{x}_i} / 2$. Thus, $\| \mb{z} - \mb{x}_i \| < \delta / 4 + \delta_{\mb{x}_i} / 2 < \delta_{\mb{x}_i}$, so $\mb{z} \in\mathcal{X}$. As $\mb{z}$ was arbitrary, $K \oplus \overline{B}_{\delta / 4} \subseteq \mathcal{X}$, so pick $\varepsilon \leq \delta / 4$. The set $K \oplus \overline{B}_{\varepsilon}$ is compact as $K\times\overline{B}_{\varepsilon}$ is compact and $(\mb{x},\mb{y})\mapsto\mb{x}+\mb{y}$ is continuous.
\end{proof}

\subsection{Proof of Theorem \ref{thm:onestep}}
\begin{proof}
Consider a compact set $K\subset\mathcal{X}$ and corresponding $h_1\in I$ and $M_K\in\R_{++}$, and fix a sample period $h\in(0,h_1)$. By Lemma \ref{lem:ballfit}, there exists an $\varepsilon\in\R_{++}$ such that the compact set $N=K\oplus\overline{B}_{\varepsilon}$ satisfies $N\subset\mathcal{X}$. 
By assumption, $\mb{k}_h$ is bounded on $K$, and $\mb{f}$ and $\mb{g}$ are bounded on $N$ since they are continuous, implying there exists an $M \in \R_{++}$ such that:
\begin{equation}
\label{eqn:dynbnd}
    \| \mb{f}(\mb{z}) + \mb{g}(\mb{z})\mb{k}_h(\mb{y}) \| \leq M, 
\end{equation}
for all $\mb{y}\in K$ and $\mb{z}\in N$. As $\mb{f}$ and $\mb{g}$ are locally Lipschitz over $\mathcal{X}$, they are globally Lipschitz over $N$. Therefore:
\begin{align}
\label{eqn:globlip}
    & \| \mb{f}(\mb{z}) + \mb{g}(\mb{z})\mb{k}_h(\mb{y}) -  (\mb{f}(\mb{y}) + \mb{g}(\mb{y})\mb{k}_h(\mb{y})) \| \\
    % &~ \leq \| \mb{f}(\mb{z}) - \mb{f}(\mb{y}) \| + \| \mb{g}(\mb{z}) - \mb{g}(\mb{y}) \| \| \mb{k}_h(\mb{y}) \| \nonumber\\
    &~ \leq (L_{\mb{f}} + L_{\mb{g}}M_K) \| \mb{z} - \mb{y} \| = \rho(\| \mb{z} - \mb{y} \|), \nonumber
\end{align}
for all $\mb{y}\in K$ and $\mb{z} \in N$, where $L_{\mb{f}}, L_{\mb{g}} \in \R_{++}$ are Lipschitz constants for $\mb{f}$ and $\mb{g}$, respectively, and $\rho \in \mathcal{K}_\infty$ satisfies $\rho(r) = (L_{\mb{f}} + L_{\mb{g}} M_K)r$ for all $r \in \R_+$. Let $\mb{x}\in K$. Then:
\begin{align}\label{eqn:consistency-add-subtract}
    &\mb{F}_h^e(\mb{x}, \mb{k}_h(\mb{x})) - \mb{F}_h^{a, p}(\mb{x}, \mb{k}_h(\mb{x}))\\
    &~= \int_0^h [\mb{f}(\bs{\varphi}(t)) + \mb{g}(\bs{\varphi}(t)) \mb{k}_h(\mb{x})] ~\mathrm{d}t \nonumber \\&\qquad\qquad\qquad\qquad\qquad - h \sum_{i = 1}^p b_i (\mb{f}(\mb{z}_i) + \mb{g}(\mb{z}_i)\mb{k}_h(\mb{x})) \nonumber\\
    &~= \int_0^h [\mb{f}(\bs{\varphi}(t)) + \mb{g}(\bs{\varphi}(t)) \mb{k}_h(\mb{x}) - (\mb{f}(\mb{x}) + \mb{g}(\mb{x})\mb{k}_h(\mb{x}))] ~\mathrm{d}t \nonumber\\
    &~~~~ + h\sum_{i = 1}^p b_i [\mb{f}(\mb{x}) + \mb{g}(\mb{x})\mb{k}_h(\mb{x}) - (\mb{f}(\mb{z}_i) + \mb{g}(\mb{z}_i)\mb{k}_h(\mb{x}))], \nonumber
\end{align}
where we make use of the fact $\sum_{i = 1}^p b_i = 1$. 

To bound the first term in \eqref{eqn:consistency-add-subtract}, let $h_2 \in (0, h_1)$ satisfy $h_2 < \varepsilon / M$. By continuity of $\bs{\varphi}$, if $\bs{\varphi}(t_0) \not\in N$ for any $t_0 \in I$, then there is a minimal time $t^* \in (0, t_0)$ such that $\| \bs{\varphi}(t) - \mb{x} \| < \varepsilon$ for all $t \in [0, t^*)$ and $\| \bs{\varphi}(t^*) - \mb{x} \| = \varepsilon$. We have:
\begin{equation}
     \| \bs{\varphi}(t) - \mb{x} \| \leq \int_0^t \| \mb{f}(\bs{\varphi}(s)) + \mb{g}(\bs{\varphi}(s))\mb{k}_h(\mb{x}) \| ~\mathrm{d}s \leq Mt,
\end{equation}
for all $t \in [0, t^*]$. Since $\varepsilon = \| \bs{\varphi}(t^*) - \mb{x} \| \leq Mt^*$, we know that $t^* \geq \varepsilon / M > h_2$. Thus if $h \in (0, h_2)$, then:
\begin{equation}
    \| \bs{\varphi}(t) - \mb{x} \| \leq Mt \leq Mh < Mh_2 < \varepsilon,
\end{equation}
for all $t\in[0,h]$, implying $\bs{\varphi}(t)\in N$ for all $t\in[0,h]$.

To bound the second term in \eqref{eqn:consistency-add-subtract}, we show by induction that if $h$ is sufficiently small, then $\mb{z}_i \in N$ for all $i \in \{ 1, \dots, p \}$. First, since $\mb{z}_1 = \mb{x}$, we have $\mb{z}_1 \in N$. Next, for $i \in \{ 1, \dots,  p \}$, suppose $\mb{z}_j \in N$ for all $j \in \{ 1, \dots, i - 1 \}$. Considering the definition of $\mb{z}_i$ in \eqref{eqn:rkz} and the bound \eqref{eqn:dynbnd}, we have that:
\begin{align}
    \| \mb{z}_i - \mb{x} \| &\leq h \sum_{j = 1}^{i - 1} | a_{i,j} | \| \mb{f}(\mb{z}_j) + \mb{g}(\mb{z}_j)\mb{k}_h(\mb{x}) \|\\
    &\leq Mh \sum_{j = 1}^{i - 1} | a_{i,j} | \leq Mh(p - 1)\max_{j, k} | a_{j,k} | \triangleq Lh, \nonumber
\end{align}
Let $h^* \in(0,h_2)$ satisfy $h^* < \varepsilon / L$. Then for $h\in(0,h^*)$, we have $\| \mb{z}_i - \mb{x} \| < \varepsilon$, or $\mb{z}_i\in N$. Since this choice of $h^*$ does not depend on $i$, we can conclude by induction that if $h \in(0,h^*)$, then $\mb{z}_i \in N$ for all $i \in \{ 1, \dots, p \}$.

We have shown that if $h\in(0,h^*)$, then $\bs{\varphi}(t)\in N$ for all $t\in[0,h]$, and $\mb{z}_i\in N$ for $i\in\{1,\ldots,p\}$. Thus, using the bound \eqref{eqn:globlip} in \eqref{eqn:consistency-add-subtract}, we have that:
\begin{align}
    &\| \mb{F}_h^e(\mb{x}, \mb{k}_h(\mb{x})) - \mb{F}_h^{a, p}(\mb{x}, \mb{k}_h(\mb{x})) \|\\
    &~ \leq \int_0^h \rho(\| \bs{\varphi}(t) - \mb{x} \|) ~\mathrm{d}t + h \sum_{i = 1}^p b_i \rho(\| \mb{z}_i - \mb{x} \|) \nonumber\\
    &~ \leq h\rho(Mh) + h \sum_{i = 1}^p b_i \rho(Lh) \leq h\tilde{\rho}(h),\nonumber
\end{align}
and $\tilde{\rho} \in \K$ satisfies $\tilde{\rho}(r) = \rho(Mr) + \rho(Lr)$ for all $r \in \R_+$.
\end{proof}

\subsection{Proof of Theorem \ref{thm:regval}}
\begin{proof}
The boundary $\partial \C$ is compact as a closed subset of the compact set $\C$. Thus, $\sigma \triangleq \min_{\mb{x} \in \partial \C} \| \nabla \safety(\mb{x}) \|_2 $ is strictly positive since $0$ is a regular value. By Lemma \ref{lem:ballfit}, there is an $\varepsilon' \in \R_{++}$ with $\C \oplus \overline{B}_{\varepsilon'} \subset \mathcal{X}$ and $\C \oplus \overline{B}_{\varepsilon'}$ compact.

Consider a state $\mb{x}\in\C\oplus\overline{B}_{\varepsilon'}$ with $\mb{x} \not\in \mathcal{C}$. There exists a $\mb{y}\in\partial\C$ such $d_\C(\mb{x}) = \Vert\mb{y}-\mb{x}\Vert_2>0$. Since $\safety$ has 0 as a regular value, by \cite[Proposition 1.1.9]{clarke2008nonsmooth} we have that:
\begin{equation}
   \grad\safety(\mb{y}) =-\Vert\grad\safety(\mb{y})\Vert_2\frac{\mb{x}-\mb{y}}{\Vert\mb{x}-\mb{y}\Vert_2},
\end{equation}
that is $\grad\safety(\mb{y})$ is anti-parallel to $\mb{x}-\mb{y}$. As $\overline{B}_{\varepsilon'}$ is convex, we have that $(1-\lambda)\mb{y}+\lambda\mb{x} \in\C\oplus\overline{B}_{\varepsilon'}$ for all $\lambda\in[0,1]$. For some $\lambda^* \in [0, 1]$, the convex combination $\bs{\xi} \triangleq (1 - \lambda^*)\mb{y} + \lambda^*\mb{x}$ satisfies:
\begin{align}
    \safety(\mb{x}) &= \safety(\mb{y}) + (\mb{x} - \mb{y})^\top \grad\safety(\mb{y})\\
    &\qquad + \frac{1}{2}(\mb{x} - \mb{y})^\top \nabla^2\safety(\bs{\xi})(\mb{x} - \mb{y})\nonumber\\
    &= -\|\nabla \safety(\mb{y})\|_2\| \mb{x} - \mb{y} \|_2 + \frac{1}{2}(\mb{x} - \mb{y})^\top \nabla^2\safety(\bs{\xi})(\mb{x} - \mb{y}).\nonumber
\end{align}

Since $\C\oplus\overline{B}_{\varepsilon'}$ is compact, there is an upper bound $\mu \in \R_+$ such that $\max_{\mb{z} \in \C \oplus \overline{B}_{\varepsilon'}} \| \nabla^2 \safety(\mb{z}) \|_2 = \mu$, and thus we have:
\begin{equation}
    \safety(\mb{x}) \leq -(\sigma - \frac{\mu}{2}\| \mb{x} - \mb{y} \|_2) \| \mb{x} - \mb{y} \|_2.
\end{equation}
If $\| \mb{x} - \mb{y} \|_2 \leq \sigma / \mu$, then:
\begin{equation}
    \safety(\mb{x}) \leq -\frac{\sigma}{2} \| \mb{x} - \mb{y} \|_2 = -\frac{\sigma}{2} d_\C(\mb{x}).
\end{equation}
We pick $\varepsilon \in \R_{++}$ such that $\varepsilon \leq \min{\{\varepsilon', \sigma/\mu \}}$, and for any $\eta \in \R_{++}$, we pick $\delta \in \R_{++}$ such that $\delta < \sigma\eta / 2$.
\end{proof}

\subsection{Proof of Theorem \ref{thm:convexity}}
\begin{proof}
For all $(\mb{x}, \mb{u}) \in \mathcal{Z}$, denote:
\begin{equation}
    \mb{F}_h^{a, p}(\mb{x}, \mb{u}) = ((\mb{F}_1)_h^{a, p}(\mb{x}, \mb{u}), \dots, (\mb{F}_\gamma)_h^{a, p}(\mb{x}, \mb{u})),
\end{equation}
where $(\mb{F}_i)_h^{a, p}:  \mathcal{Z} \to \R^\ell$ for all $i \in \{ 1, \dots, \gamma \}$. For $(\mb{x}, \mb{u}) \in \mathcal{Z}$, the block vector $\mb{r}(\mb{x}, \mb{u})$ can be nonzero only in the last ($\gamma$th) block. Noting the block chain-of-integrators structure of $\mb{A}$, for any degree $d \in \{ 0, \dots, \gamma - 1 \}$, the block vector $\mb{A}^d\mb{r}(\mb{x}, \mb{u})$ can be nonzero only in the $(\gamma - d)$th block, and for a degree $d$ polynomial $\rho_d$, the block vector $\rho_d(\mb{A})\mb{r}(\mb{x}, \mb{u})$ can be nonzero only in the last $d + 1$ blocks (that is, blocks $\gamma - d$ through $\gamma$).

Consider a state-input pair $(\mb{x}, \mb{u}) \in \mathcal{Z}$. We have:
\begin{align}
    \mb{F}_h^{a, p}(\mb{x}, \mb{u}) &= \mb{x} + h \sum_{i = 1}^p b_i ( \mb{A}\mb{z}_i + \mb{r}(\mb{z}_i, \mb{u}) ), \label{eqn:proofapproxmap}\\ \label{eqn:zidef}
    \mb{z}_i &= \mb{x} + h \sum_{j = 1}^{i - 1} a_{i, j} (\mb{A}\mb{z}_j + \mb{r}(\mb{z}_j, \mb{u})),
\end{align}
with $\mb{z}_1 = \mb{x}$. By induction, for any $i \in \{ 1, \dots, p \}$, we show:
\begin{equation}\label{eqn:induction}
    \mb{z}_i = \rho_{i, i - 1}(\mb{A})\mb{x} + \sum_{j = 1}^{i - 1} \sigma_{i, i - j - 1}(\mb{A})\mb{r}(\mb{z}_j, \mb{u}),
\end{equation}
where $\rho_{i, i - 1}$ is a degree $i - 1$ polynomial, and for $j \in \{ 1, \dots, i - 1 \}$, $\sigma_{i, i - j - 1}$ is a degree $i - j - 1$ polynomial. Indeed, $\mb{z}_1 = \mb{I}\cdot\mb{x}$, and assuming \eqref{eqn:induction} holds for $0, \dots, i - 1$, substituting \eqref{eqn:induction} into \eqref{eqn:zidef} yields the following:
\begin{align}
    \mb{z}_i &= \underbrace{\overbrace{\Bigg(\mb{I} + h\sum_{j = 1}^{i - 1} a_{i, j} \overbrace{\mb{A}\rho_{j, j - 1}(\mb{A})}^{\mathrm{degree}~j}\Bigg)}^{\mathrm{degree}~i - 1}}_{\triangleq \rho_{i, i - 1}(\mb{A})}\mb{x} + h\sum_{j = 1}^{i - 1} a_{i, j} \mb{r}(\mb{z}_j, \mb{u})\nonumber\\
    &\quad~ + h\sum_{k = 1}^{i - 1} \sum_{j = 1}^{k - 1} a_{i, k} \mb{A}\sigma_{k, k - j - 1}(\mb{A})\mb{r}(\mb{z}_j, \mb{u}),
\end{align}
which we may further manipulate to obtain:
\begin{align}
    & \mb{z}_i - \rho_{i, i - 1}(\mb{A})\mb{x} \nonumber\\
    & = \sum_{j = 1}^{i - 1} \underbrace{h\Bigg( a_{i, j} + \sum_{k = j + 1}^{i - 1} a_{i, k}\underbrace{\mb{A}\sigma_{k, k - j - 1}(\mb{A})}_{\mathrm{degree}~ k - j} \Bigg)}_{\mathrm{degree}~ i - j - 1} \mb{r}(\mb{z}_j, \mb{u}),\\
    &~\triangleq \sum_{j = 1}^{i - 1} \sigma_{i, i - j - 1}(\mb{A})\mb{r}(\mb{z}_j, \mb{u}),
\end{align}
establishing \eqref{eqn:induction} holds for $i$. Substituting the expression \eqref{eqn:induction} into \eqref{eqn:proofapproxmap} and following a similar sequence of steps, we find a degree $p$ polynomial $\tilde{\rho}_p$, and for each $i \in \{ 1, \dots, p \}$, a degree $p - i$ polynomial $\tilde{\sigma}_{p - i}$ such that:
\begin{equation}
    \mb{F}_h^{a, p}(\mb{x}, \mb{u}) = \tilde{\rho}_{p}(\mb{A})\mb{x} + \sum_{i = 1}^p \tilde{\sigma}_{p - i}(\mb{A})\mb{r}(\mb{z}_i, \mb{u}).
\end{equation}
For $i \in \{ 1, \dots, p \}$, the term $\tilde{\sigma}_{p - i}(\mb{A})\mb{r}(\mb{z}_i, \mb{u})$ can be nonzero only in blocks $\gamma - (p - i) = q + i - 1$ through $\gamma$. The highest-order polynomial multiplying the block vectors $\mb{r}(\mb{z}_1, \mb{u}), \dots, \mb{r}(\mb{z}_p, \mb{u})$ is $\tilde{\sigma}_{p - 1} = \tilde{\sigma}_{\gamma - q}$. Therefore, the functions $(\mb{F}_1)_h^{a, p}, \dots, (\mb{F}_{q - 1})_h^{a, p}$ are independent of their second argument (they depend only on state). Moreover, $(\mb{F}_q)_h^{a, p}(\mb{x}, \mb{u})$ depends on the block vector $\mb{r}(\mb{z}_1, \mb{u}) = \mb{r}(\mb{x}, \mb{u})$, which depends on $\mb{u}$ affinely, and does not depend on the block vectors $\mb{r}(\mb{z}_2, \mb{u}), \dots, \mb{r}(\mb{z}_p, \mb{u})$, which may depend on $\mb{u}$ nonlinearly.

The composition $\safety_h\circ\mb{F}_h^{a,p}:\mathcal{Z}\to\R$ satisfies:
\begin{equation}
    \safety_h(\mb{F}_h^{a,p}(\mb{x},\mb{u})) = \tilde{\safety}_h((\mb{F}_1)^{a,p}_h(\mb{x},\mb{u}),\cdots,(\mb{F}_q)^{a,p}_h(\mb{x},\mb{u})), \nonumber
\end{equation}
for all $(\mb{x}, \mb{u}) \in \mathcal{Z}$. The composition of concave and affine functions is concave, so $\safety_h\circ \mb{F}^{a,p}_h$ is concave in its second argument, and $\phi_h$ in \eqref{eqn:cvxcnstrnt} is convex in its second argument.
\end{proof}

\subsection{Proof of Theorem \ref{thm:optprob}}
\begin{proof}
Consider any $h\in(0,h^*)$ and $\mb{x}\in\mathcal{X}$. As $\safety_h$ is a SD-CBF on $\C$, there exists a $\mb{u}'\in\R^m$ such that $(\mb{x},\mb{u}')\in\mathcal{Z}$ and:
\begin{equation}
    \safety_h(\mb{F}_h^{a,p}(\mb{x},\mb{u}'))-\safety_h(\mb{x}) \geq -h\alpha(\safety_h(\mb{x})),
\end{equation}
implying that $\mb{u}'\in\mathcal{F}(\mb{x})$. Thus the optimization problem in \eqref{eq:sdcbf Controller} is feasible. Define the compact, convex set:
\begin{equation}
    A = \left\{\mb{u}\in\R^m ~|~ \Vert\mb{u}-\mb{k}_d(\mb{x})\Vert_2^2 \leq \Vert\mb{u'}-\mb{k}_d(\mb{x})\Vert_2^2 \right\}.
\end{equation}
Note that $\mb{u}'\in A$. As the set $\mathcal{F}(\mb{x})$ is closed and convex, the set $A\cap\mathcal{F}(\mb{x})$ is compact, convex, and non-empty. As the cost is continuous and strictly convex with respect to $\mb{u}$, there is a unique minimizer $\mb{u}^*\in A\cap{\mathcal{F}}(\mb{x})$. We have $\Vert\mb{u}^*-\mb{k}_d(\mb{x})\Vert_2^2 \leq \Vert\mb{u}'-\mb{k}_d(\mb{x})\Vert_2^2  < \Vert\mb{u}-\mb{k}_d(\mb{x})\Vert_2^2$ for all $\mb{u}\in\mathcal{F}(\mb{x})\setminus A$, implying $\mb{u}^*$ is the unique minimizer in $\mathcal{F}(\mb{x})$. Thus:
\begin{equation}
    \safety_h(\mb{F}_h^{a,p}(\mb{x},\mb{k}_h(\mb{x})))-\safety_h(\mb{x}) \geq -h\alpha(\safety_h(\mb{x})),
\end{equation}
and as $\mb{x}$ and $h$ were arbitrary, we have that $\{\safety_h~|~h\in I\}$ is a family of SD-BFs on $\C$ for the family $(\mb{k}_h,\mb{F}^{a,p}_h)$.
\end{proof}

% %%%%%%%%%% Bibliography %%%%%%%%%%%%%%%%%%%%%%%
\bibliographystyle{IEEEtran} 
\bibliography{main}

\end{document}

%% file: preamble.tex
\usepackage{cite}
\usepackage{times}
\usepackage{setspace}
\usepackage{url}
\spacing{1}
\usepackage[utf8]{inputenc}
\usepackage[T1]{fontenc}
\usepackage{graphicx}		% For \includegraphics
\usepackage{wrapfig}
\usepackage[format=plain,font=footnotesize,labelfont=bf,labelsep=period]{caption}
\usepackage{sidecap} % For sidecaptions for figures
\usepackage{subfig}
\usepackage[export]{adjustbox}
\usepackage[font=small]{caption}
\usepackage{float}
\usepackage{dblfloatfix}

\usepackage{amsmath} % assumes amsmath package installed
\usepackage{amssymb}  % assumes amsmath package installed
\usepackage{amsthm}
\usepackage{mathtools}
\usepackage[normalem]{ulem}
\usepackage{paralist}	% For enumerations with better titled numbers
\usepackage[space]{grffile} % Space in filenames
\usepackage{color}
\usepackage{bbm}
\usepackage{placeins} % Floatbarrier
\usepackage{array}
\usepackage{siunitx}
\usepackage{hyperref}
\usepackage{enumitem}

\newtheorem{theorem}{Theorem}

\newtheorem{lemma}{Lemma}
\theoremstyle{definition}
\newtheorem{definition}{Definition}
\theoremstyle{remark}
\newtheorem{remark}{Remark}
\theoremstyle{definition}

\theoremstyle{definition}

\newcommand{\N}{\mathbb{N}}

\newcommand{\R}{\mathbb{R}}
\newcommand{\C}{\mathcal{C}}

\newcommand{\K}{\mathcal{K}}

\definecolor{darkblue}{RGB}{0,0,102}
\definecolor{lightblue}{RGB}{77,77,148}

\definecolor{gold}{RGB}{234, 170, 0}
\definecolor{metallic_gold}{RGB}{139, 111, 78}

\newcommand{\mb}[1]{\mathbf{ #1 }}
\newcommand{\bs}[1]{\boldsymbol{ #1 }}
\newcommand{\bb}[1]{\mathbb{ #1 }}

\newcommand{\grad}{\nabla}

\DeclareMathOperator*{\argmin}{argmin}

\newcommand{\new}[1]{{\color{black} #1}}

\renewcommand{\baselinestretch}{.985}
\allowdisplaybreaks

\newcommand{\lmat}{\begin{bmatrix}}
\newcommand{\rmat}{\end{bmatrix}}